\documentclass[a4paper,twocolumn,11pt,accepted=2024-11-06,aps]{quantumarticle}
\pdfoutput=1
\usepackage[utf8]{inputenc}
\usepackage[english]{babel}
\usepackage[T1]{fontenc}
\usepackage{amsmath}
\usepackage{hyperref}

\usepackage[numbers,compress]{natbib}

\usepackage{tikz}

\usepackage{braket}
\usepackage{amsmath,amssymb,amsfonts}
\usepackage{amsthm}
\usepackage{mathtools}
\usepackage{graphicx}
\usepackage{xcolor}
\usepackage{framed}

\usepackage[capitalize,nameinlink]{cleveref}
\usepackage{comment}

\newcommand{\be}{\begin{equation}}\newcommand{\ee}{\end{equation}}\newcommand{\ba}{\begin{eqnarray}}\newcommand{\ea}{\end{eqnarray}}\newcommand{\ban}{\begin{eqnarray*}}\newcommand{\ean}{\end{eqnarray*}}

\newtheorem{lemma}{Lemma}

\theoremstyle{definition}

\theoremstyle{remark}

\newcommand{\mM}{\mathcal{M}}

\def\openone{\mathds{1}}

\newcommand{\tS}{\widetilde{S}}

\newcommand{\D}{\mathcal{D}}

\newcommand{\ketbra}[1]{|#1\rangle\!\!\!\; \langle #1 |}
\newcommand{\tr}[1]{\operatorname{Tr}\!\left[#1\right]}

\usepackage{dsfont}
\usepackage{enumitem}

\newcommand{\Tr}{\mathrm{Tr}}

\newcommand{\mMR}{\widetilde{\mathcal{M}}}
\newcommand{\DBS}{D_{\rm BS}}

\newcommand{\numU}{1}
\newcommand{\numDQ}{2}
\newcommand{\numBS}{3}

\newcommand{\SigmaU}{\Sigma_{\mathsf{M},\gamma}^{(\numU)}}
\newcommand{\SUmegaki}{S_{\mathsf{M},\gamma}^{(\numU)}}
\newcommand{\SigmaBS}{\Sigma_{\mathsf{M},\gamma}^{(\numBS)}}
\newcommand{\SBS}{S_{\mathsf{M},\gamma}^{(\numBS)}}
\newcommand{\SigmaDQ}{\Sigma_{\mathsf{M},\gamma}^{(\numDQ)}}
\newcommand{\SDQ}{S_{\mathsf{M},\gamma}^{(\numDQ)}}
\newcommand{\Scl}{S^{\textrm{clax}}_{\mathsf{M},\gamma}}

\begin{document}

\title{Observational entropy with general quantum priors}

\author{Ge Bai}
\affiliation{Centre for Quantum Technologies, National University of Singapore, 3 Science Drive 2, Singapore 117543}
\orcid{0000-0002-6814-8840}

\author{Dominik \v{S}afr\'anek}
\affiliation{Center for Theoretical Physics of Complex Systems, Institute for Basic Science (IBS), Daejeon 34126, Korea}
\affiliation{Basic Science Program, Korea University of Science and Technology (UST), Daejeon - 34113, Korea}
\orcid{0000-0002-6861-395X}

\author{Joseph Schindler}
\affiliation{F\'{\i}sica Te\`{o}rica: Informaci\'{o} i Fen\`{o}mens Qu\`{a}ntics, Departament de F\'{\i}sica, Universitat Aut\`{o}noma de Barcelona, 08193 Bellaterra, Spain}
\orcid{0000-0002-8799-9800}

\author{Francesco Buscemi}\email{buscemi@nagoya-u.jp}
\affiliation{Department of Mathematical Informatics, Nagoya University, Furo-cho Chikusa-ku, Nagoya 464-8601, Japan}
\orcid{0000-0001-9741-0628}

\author{Valerio Scarani}\email{physv@nus.edu.sg}
\affiliation{Centre for Quantum Technologies, National University of Singapore, 3 Science Drive 2, Singapore 117543}
\affiliation{Department of Physics, National University of Singapore, 2 Science Drive 3, Singapore 117542}
\orcid{0000-0001-5594-5616}

\begin{abstract}
Observational entropy captures both the intrinsic uncertainty of a thermodynamic state and the lack of knowledge due to coarse-graining.  We demonstrate two interpretations of observational entropy, one as the statistical deficiency resulting from a measurement, the other as the difficulty of inferring the input state from the measurement statistics by quantum Bayesian retrodiction. These interpretations show that the observational entropy implicitly includes a uniform reference prior. Since the uniform prior cannot be used when the system is infinite-dimensional or otherwise energy-constrained, we propose generalizations by replacing the uniform prior with arbitrary quantum states that may not even commute with the state of the system. We propose three candidates for this generalization, discuss their properties, and show that one of them gives a unified expression that relates both interpretations. 
\end{abstract}

\maketitle

\section{Introduction}

A few pages after defining the entropy that nowadays bears his name, von Neumann warns the reader that the quantity that he just defined is, in fact, unable to capture the phenomenological behavior of thermodynamic entropy~\cite{von1955mathematical}. More precisely, while the von Neumann entropy $S(\rho):=-\tr{\rho\ln\rho}$ is always invariant in a closed system as a consequence of its invariance under unitary evolutions, the thermodynamic entropy of a closed system can instead increase, as it happens for example in the free expansion of an ideal gas. The explanation that von Neumann gives for this apparent paradox is the following: thermodynamic entropy includes not only the intrinsic ignorance associated with the microscopic state $\rho$ of the system, but also the lack of knowledge arising from a macroscopic coarse-graining of it. The latter lack of knowledge becomes worse as the gas expands. This observation leads him to introduce an alternative quantity, that he calls \textit{macroscopic entropy}, for which an $H$-theorem can be proved~\cite{vonNeumann1929translation}.

In recent years, von Neumann's macroscopic entropy and a generalization thereof called \textit{observational entropy} (OE) has been the object of renewed interest~\cite{safranek2019a,safranek2019b,safranek2021brief,SW21,safranek2021generalized,buscemi2022observational}, finding a number of applications~\cite{riera2020finite,safranek2020classical,deutsch2020probabilistic,faiez2020typical,nation2020snapshots,strasberg2021clausius,hamazaki2022speed,modak2022observational,sreeram2023witnessing,schindler2023continuity,safranek2023work,safranek2023measuring,safranek2023ergotropic}. So far, even when the narrative is based on a quantum state $\rho$ being subject to a measurement $\mathsf{M}$, all the definitions fit in classical stochastic thermodynamics.

In this paper, we explore possible generalizations of OE. We note that the original OE includes an implicit prior belief about the state, which is the uniform distribution. Since in several applications the uniform prior cannot be used, e.g., in infinite-dimensional or continuous variable systems, or does not play well with other physical constraints, e.g., in thermodynamic systems with a nondegenerate Hamiltonian at finite temperature, we allow the observer to have a non-uniform prior. More generally, we consider the possibility that the observer has a reference prior described by an arbitrary density operator, which may not even commute with the state of the system. In this case, classical probability distributions may not be sufficient to describe the non-commutativity between the state and the reference, and thus the original definition of OE is not applicable.

\section{Classical OE and reference states}

In what follows, we restrict our attention to finite-dimensional quantum systems, with Hilbert space $\mathbb{C}^d$, and finite measurements, i.e., positive operator-valued measures (POVMs) $\mathsf{M}=\{\Pi_y\}_{y}$ labeled by the elements of a finite set $y\in\{1,\dots,m\}$. In this context, the definition of OE is
\ba\label{eq:original}
S_{\mathsf{M}}(\rho)\coloneqq-\sum_{y=1}^m p_y\,\ln\frac{p_y}{V_y}\;,
\ea
where $p_y:=\tr{\rho \Pi_y}$ and $V_y:=\tr{\Pi_y}$. One of the conceptual advantages of OE is that it is able to ``interpolate'' between Boltzmann and Gibbs--Shannon entropies. On the one hand, if the measurement is so coarse-grained that one of its elements (say $\Pi_1$) is the projector on the support of $\rho$, then $S_{\mathsf{M}}(\rho)=\ln V_1$ takes the form of a Boltzmann entropy. If, on the other hand, the measurement is projective and rank-one (i.e., $V_y=1$ for all $y$), then $S_{\mathsf{M}}(\rho)$ coincides with the Shannon entropy of the probability distribution $\{p_y\}$, which is equal to $S(\rho)$ when $\rho=\sum_y p_y \Pi_y$.

In general, it holds that~\cite{safranek2019b}
\begin{align}\label{Sigma0}
    \Sigma_{\mathsf{M}}(\rho) \coloneqq S_{\mathsf{M}}(\rho)-S(\rho)\ge 0\,.
\end{align}
If $S(\rho)$ represents, in von Neumann's original narrative, the \textit{least} uncertainty that an observer, able to perform \textit{any} measurement in principle allowed by quantum theory, has about the state of the system, then the additional uncertainty $\Sigma_{\mathsf{M}}(\rho)$ included in OE is a consequence of observing the system through the ``lens'' provided by the given measurement $\mathsf{M}$. Thus, in this sense, OE can be seen as a measure of how inadequate a given measurement $\mathsf{M}$ is with respect to the state $\rho$.

\subsection{OE from statistical deficiency}\label{sec:OE-deficiency}

The above discussion suggests one possible generalization of OE, starting from the re-writing of~\eqref{Sigma0} recently noticed by some of us~\cite{buscemi2022observational}. Consider the \textit{measurement channel}
$\mM$ associated to the measurement $\mathsf{M}$, defined as
\ba \label{eq:measurement_channel}
\mM(\rho) \coloneqq \sum_y \Tr[\Pi_y \rho] \ketbra{y}\;,
\ea
where $\{\ket{y}\}$ is an arbitrary but fixed orthonormal basis of the system that records the measurement outcome. By further noticing that $V_y=d\tr{\Pi_y u}$ with $u=\openone/d$ the maximally mixed state, one obtains 
\begin{align}\label{SigmaBSS}
    \Sigma_{\mathsf{M}}(\rho) = D(\rho\|u)-D(\mathcal{M}(\rho)\|\mathcal{M}(u))\;,
\end{align}
where
\begin{align} \label{eq:D_Umegaki}
    D(\rho\|\sigma) \coloneqq \Tr[\rho(\ln\rho-\ln\sigma)]
\end{align} is the Umegaki quantum relative entropy between states $\rho$ and $\sigma$ \cite{umegaki,umegaki1962conditional}, which generalizes the relative entropy (a.k.a.~Kullback-Leibler divergence)
\begin{align}\label{eq:KL}
    D(p\|q) \coloneqq \sum_i p_i \ln \frac{p_i}{q_i}
\end{align} between probability distributions $\{p_i\}$ and $\{q_i\}$ \cite{kullback1951}. 

The expression~\eqref{SigmaBSS} makes it clear that the quantity $\Sigma_{\mathsf{M}}(\rho)$ exactly equals the loss of distinguishability between the signal $\rho$ and the totally uniform background $u$ that occurs when the measurement $\mathsf{M}$ is used instead of the best possible measurement allowed by quantum theory. In statistical jargon, we thus say that $\Sigma_{\mathsf{M}}(\rho)$ measures the \textit{statistical deficiency} of the measurement $\mathsf{M}$ in distinguishing $\rho$ against $u$.

This observation enlightens something implicit in the original definition \eqref{eq:original} of OE: the coarse-graining is captured by the ``volumes'' $V_y=\tr{\Pi_y}$ only because the maximally mixed state is chosen as the reference background. It is thus natural to try to incorporate more general references in the definition of OE. A direct generalization could be obtained, therefore, by replacing $u$ with another reference state $\gamma$ in~\eqref{SigmaBSS}, so that
\begin{align}\nonumber
S_{\mathsf{M},\gamma}(\rho) &\coloneqq S(\rho) + \Sigma_{\mathsf{M},\gamma}(\rho)\\
    &\coloneqq S(\rho)+\mathfrak{D}(\rho\|\gamma)-D(\mathcal{M}(\rho)\|\mathcal{M}(\gamma)) \,, \label{eq:sigma_rho_gamma_initial}
\end{align}
where $\mathfrak{D}(\cdot\|\cdot)$ represents some non-commutative generalization of the Kullback--Leibler divergence, not necessarily Umegaki's one.

\subsection{OE from irretrodictability}

There exists another evocative re-writing of \eqref{Sigma0}, which in turn suggests that further structures may play a role in the definition of OE. Specifically, here we exhibit a \textit{dynamical} interpretation of OE, based on a measurement process defined as follows.

Let $\rho=\sum_{x=1}^d\lambda_x\ketbra{\psi_x}$ be a diagonal decomposition of the state of the system. We consider a stochastic map associated to a prepare-and-measure protocol: with probability $\lambda_x$, the state $\ket{\psi_x}$ is prepared, and it is then measured with the POVM $\mathsf{M}=\{\Pi_y\}_y$, yielding outcome $y$ with a probability given by the Born rule, that is
\begin{align}
P_F(y|x) &\coloneqq \bra{\psi_x}\Pi_y\ket{\psi_x}\,,\\
P_F(x,y) &\coloneqq \lambda_x\bra{\psi_x}\Pi_y\ket{\psi_x}\,.
\label{eq:forward-prepare-and-measure-process}
\end{align}
The subscript $F$ stands for ``forward''. This is because, as we will see in what follows, the quantity $\Sigma_{\mathsf{M}}(\rho)$ in~\eqref{Sigma0} emerges also from a comparison between the forward map defined above and a suitably defined ``reverse'' map.

Traditionally, the definition of the reverse process, given a forward process, relies on a detailed knowledge of the physical dynamics involved~\cite{crooks-theorem,esposito-RMP-review,landi-RMP-review}. In the absence of such knowledge, which is typically the case for a system interacting with a complex environment, one must resort to physical intuition, plausibility arguments, or, failing that, arbitrary assumptions. In order to avoid all this, a systematic recipe has recently been found~\cite{BS21,AwBS}, which allows to define the reverse process only from unavoidable rules of \textit{logical retrodiction}: specifically, Jeffrey's theory of probability kinematics~\cite{jeffrey} or, equivalently, Pearl's virtual evidence method~\cite{pearl,CHAN200567}.

The idea is as follows: given a forward conditional probability $P_F(y|x)$, how should information, obtained at later times about the final outcome and encoded in an arbitrary distribution $q_y$, be propagated back to the initial state in a way that ensures logical consistency? Jeffrey's theory of probability kinematics, which is equivalent to Pearl's virtual evidence method~\cite{pearl,CHAN200567}, stipulates that the only logically consistent back-propagation rule is what is now known as \textit{Jeffrey's update}: starting from an arbitrarily chosen reference prior on the initial state $x$, say $\gamma_x$, one constructs the Bayesian inverse of $P_F(y|x)$, i.e.,
\begin{align}\label{eq:retrodictive-channel}
    P_R^{\gamma}(x|y):=\frac{\gamma_xP_F(y|x)}{\sum_{x'}\gamma_{x'} P_F(y|{x'})}\,,
\end{align}
and uses that as a \textit{stochastic channel} to back-propagate the new information $q_y$ from the final outcome $y$ to the initial state $x$, so that as the reverse process we obtain
\begin{align}\label{eq:clax-retro}
 P_R^{\gamma}(y,x)=q_y\,P_R^{\gamma}(x|y)\;.   
\end{align}
Jeffrey's update constitutes a generalization of Bayes' theorem, as the latter is recovered as a special case of the former when $q_y=\delta_{y,y_0}$, i.e., when the information about the final outcome is definite~\cite{CHAN200567}.

An important point to emphasize here is that the reference prior $\gamma_x$ used to construct the retrodictive channel in~\eqref{eq:retrodictive-channel} is merely a formal device needed to establish a mathematical correspondence between forward and backward process: it need not be related in any way with the ``true'' distribution $\lambda_x$. Likewise, the distribution $q_y$ represents \textit{new and completely arbitrary} information, which need not correspond to any input distribution under $P_F$, that is, there may exist no distribution $q'_x$ such that $q_y=\sum_xP_F(y|x)q'_x$. Moreover, in principle, $q_y$ may also be incompatible with the reference prior $\gamma_x$, in the sense that it could happen that, for some $y$, $q_y>0$ but $\sum_x\gamma_xP_F(y|x)=0$. In such a situation, one would conclude that the data \textit{falsify} the inferential model, but for simplicity we will avoid such cases by assuming that all probabilities are strictly greater than zero (though possibly arbitrarily small).

We now go back to our specific forward process~\eqref{eq:forward-prepare-and-measure-process}, i.e., $P_F(y|x)=\bra{\psi_x}\Pi_y\ket{\psi_x}$. If we choose as reference the uniform distribution $\gamma=u$, i.e.,~$\gamma_x=1/d$ for all $x$, and as new information information the outcomes' expected probability of occurrence, i.e., $q_y=p_y=\Tr[\rho\Pi_y]$, by direct substitution in~\eqref{eq:retrodictive-channel} and~\eqref{eq:clax-retro}, we obtain
\begin{align}\label{eq:reverse-simple}
     P^u_R(y,x) = p_y\bra{\psi_x}\frac{\Pi_y}{\tr{\Pi_y}}\ket{\psi_x}\;.
\end{align}
The above can also be read as a prepare-and-measure process, in which the state $\sigma_y\coloneqq \frac{\Pi_y}{\tr{\Pi_y}}$ is prepared with probability $p_y$, and later measured in the basis $\{\ket{\psi_x}\}$. The process in~\eqref{eq:reverse-simple} is the process that a retrodictive agent would infer, knowing only the forward process~\eqref{eq:forward-prepare-and-measure-process} and the outcome distribution $p_y$, but completely ignoring the actual distribution $\lambda_x$, so that the latter is replaced by the uniform distribution.

Using \eqref{eq:forward-prepare-and-measure-process} and \eqref{eq:reverse-simple}, it is straightforward to check that \begin{align}\label{eq:classical_D=S}
    \Sigma_{\mathsf{M}}(\rho)=D(P_F\|P^u_R)\;.
\end{align}
The above relation suggests an alternative interpretation for the difference $\Sigma_{\mathsf{M}}(\rho)$, as the degree of statistical distinguishability between a predictive process, i.e., $P_F(x,y)$, and a retrodictive process constructed from a uniform reference, i.e., $P_R^u(y,x)$. Thus, the larger $\Sigma_{\mathsf{M}}(\rho)$, the more \textit{irretrodictable} the process becomes~\cite{watanabe55,watanabe65}. 

Eq.~\eqref{eq:classical_D=S} also offers an alternative way of thinking about generalizations of OE, where the uniform reference is again replaced by an arbitrary state, as was done in Section~\ref{sec:OE-deficiency}, but this time for the purpose of constructing another reverse process. That is, one could also consider generalizations such as
\begin{align}\label{eq:Sj}
    \tS_{\mathsf{M},\gamma}(\rho)\coloneqq S(\rho)+\mathfrak{D}(Q_F\|Q_R^\gamma)\;,
\end{align}
where $\mathfrak{D}$ is again some quantum relative entropy (not necessarily Umegaki's), $Q_F$ is an input-output description of the quantum process consisting in preparing the state $\rho$ and measuring it with the channel $\mathcal{M}$, and $Q_R^\gamma$ is the description of the corresponding reverse process computed with respect to the reference prior $\gamma$. All these ingredients will be rigorously defined in Section~\ref{sec:IV}.

\section{A definition of OE for priors $\gamma$ such that $[\rho,\gamma]=0$}

As we have seen, in the case of conventional OE, the statistical deficiency approach and the irretrodictability approach coincide, i.e.
\begin{align*}
    \Sigma_{\mathsf{M}}(\rho)&=D(\rho\|u)-D(\mathcal{M}(\rho)\|\mathcal{M}(u))=D(P_F\|P_R^u)\,.
\end{align*}
The same holds for any prior $\gamma$ that commutes with $\rho$. Indeed, assuming $[\rho,\gamma]=0$, let us write $\gamma=\sum_{x=1}^d\gamma_x\ketbra{\psi_x}$ using the same vectors that diagonalize $\rho$. The reverse process of $P_F$ [Eq.~\eqref{eq:forward-prepare-and-measure-process}] becomes
\begin{align}\label{eq:retrodictive-channel-commutative-gamma}
    P_R^{\gamma}(x|y)&\coloneqq\frac{\gamma_x \bra{\psi_x} \Pi_y \ket{\psi_x}}{\Tr[\Pi_y \gamma]} \\
    P_R^{\gamma}(x,y)&= p_y P_R^{\gamma}(x|y)\nonumber\,,
\end{align}
and it is straightforward to verify that
\begin{align} \label{eq:D-D=DP}
    D(\rho\|\gamma)-D(\mathcal{M}(\rho)\|\mathcal{M}(\gamma)) = D(P_F||P_R^\gamma) \,.
\end{align}
Therefore, when $[\rho,\gamma]=0$, the expression
\begin{align}
    &\Scl(\rho)\nonumber\\&\coloneqq\;S(\rho)+D(\rho\|\gamma)-D(\mathcal{M}(\rho)\|\mathcal{M}(\gamma))\label{Scl1}\\
    &=\;S(\rho)+D(P_F||P_R^\gamma)\label{Scl2}\\
    &=-\Tr[\rho\;\ln\gamma]-D(\mathcal{M}(\rho)\|\mathcal{M}(\gamma))\label{eq:clax-gen-OE}\;.
\end{align}
is a generalized OE that fits both the statistical deficiency approach and the irretrodictability approach. As we are going to see, it is not obvious to ensure both interpretations when $[\rho,\gamma]\neq 0$.

Note that nothing has been said about the measurement $\mathsf{M}$, which may well not commute with either $\rho$ or $\gamma$. As a case study, let us look closely to the \textit{fully classical case} in which, besides having $[\rho,\gamma]=0$, the measurement is a projective measurement in their same eigenbasis $\{\ket{k}\}$:
\begin{align}
    \mathsf{M}&=\{\Pi_y\}_{y}\quad\textrm{ with }\quad \Pi_y=\sum_{k\in K(y)}\ket{k}\bra{k}
\end{align} where the index sets $K(y)$ are disjoint and complete as to form a POVM, i.e.~$\bigcup_y K(y)={1,...,d}$, and $K(y)\cap K(y')=\emptyset$ for $y\neq y'$. Then, denoting by $\{r_k\}$ and $\{g_k\}$ the eigenvalues of $\rho$ and $\gamma$ respectively, we have $p_y=\sum_{k\in K(y)} r_k$ and $\tr{\gamma\Pi_y}\equiv G_y=\sum_{k\in K(y)} g_k$, and Eq.~\eqref{eq:clax-gen-OE} yields 
\begin{align}
    \Scl(\rho)&=-\sum_k r_k\log g_k-D(\{p_y\}||\{G_y\})\,.
\end{align} While the second term depends only on the observed statistics $\{p_y\}$ by construction, the first term depends in general on the full information $\{r_k$\}. In fact, $\Scl$ depends only on the $\{p_y\}$ if and only if the $g_k$ are uniform in each $K(y)$ subspace; that is, $g_k=G_y/|K(y)|$ for every $k\in K(y)$ and for every $y$, or, equivalently, $\gamma=\sum_y G_y\Pi_y/|K(y)|$. In this case, $\Scl(\rho)=S_{\mathsf{M}}(\rho)$: the weights $G_y$ of the prior do not matter. The interpretation of this observation is clear: the observation gives precisely the weights $p_y$ to be attributed to each $\Pi_y$, trumping any prior belief on those weights.

When the dependence on the $\{r_k\}$ is present, it is a mild one: for instance, in the paradigmatic case where $\gamma$ is thermal, the term $-\sum_k r_k\log g_k$ is (up to an additive constant) the average energy, often assumed as known in thermodynamics. A purely ``observational'' character could be recovered with minor modifications of the definition; we leave this aside.

\section{Definitions of OE with an arbitrary quantum reference prior}\label{sec:IV}

In this section, we introduce the mathematical notations and backgrounds, and propose some candidates for OE with a general reference prior state.

\subsection{Input-output description of quantum processes}

Let $A$ and $B$ respectively denote the input and output systems of the measurement channel $\mM$ in Eq.~\eqref{eq:measurement_channel}. The general recipe for the retrodiction of a quantum process $\mM$
\cite{BS21,buscemi2022observational,Parzygnat2023axiomsretrodiction}
is defined via its Petz recovery map \cite{petz1,petz}  as
\ba \label{eq:reverse_process}
\mMR^\gamma(\tau) \coloneqq \gamma^{1/2} \mM^\dag(\mM(\gamma)^{-1/2} \tau \mM(\gamma)^{-1/2} ) \gamma^{1/2}
\ea
where $\tau := \sum_y q_y\ketbra{y}$ encodes the distribution $\{q_y\}$, describing the retrodictor's knowledge, cfr. Eq.~\eqref{eq:clax-retro}. We will mainly discuss the case where $\tau = \mM(\rho)$, namely $q_y=p_y=\Tr[\Pi_y \rho]$. For the measurement channel given in \eqref{eq:measurement_channel}, the Petz recovery map can be written as
\begin{align}
    \mMR^\gamma(\tau) = \sum_y \frac{\bra{y}\tau\ket{y}}{\Tr[\Pi_y \gamma]} \sqrt{\gamma}\Pi_y\sqrt{\gamma} \,.
\end{align}

As an ingredient for later constructions, we introduce the Choi operator \cite{choi1975completely}, defined for the process $\mM$ from system $A$ to system $B$  as 
\begin{align}\label{eq:direct-choi}
C_{\mM} &\coloneqq \sum_{i,j} \mM(\ket{i}\bra{j}) \otimes \ket{i}\bra{j} \,,
\end{align}
where $\ket{i}$ and $\ket{j}$ belong to an arbitrary but fixed orthonormal basis of the input Hilbert space $\mathcal{H}_A$ of system $A$. The reverse process has the following Choi operator
\begin{align}\label{eq:reverse-choi}
    C_{\mMR^\gamma} &\coloneqq \sum_{k,l}\ket{k}\bra{l} \otimes \mMR^\gamma(\ket{k}\bra{l}) \,,
\end{align}
where $\ket{k}$ and $\ket{l}$ belong to an arbitrarily fixed orthonormal basis of the Hilbert space $\mathcal{H}_B$.
Note that we put system $B$ first and system $A$ second, in order to have the same ordering of systems for both $C_{\mM}$ and $C_{\mMR^\gamma}$.

With such a definition, the Choi operators of the forward and reverse processes are related by the following lemma (proof in \cref{app:proof_choi}):
\begin{lemma} \label{lem:choi}
    For a quantum channel $\mM$ and its Petz map $\mMR^\gamma$, their Choi operators $C_{\mM}$ and $C_{\mMR^\gamma}$ are related as
    \begin{align} 
    &C_{\mMR^\gamma}^T=\label{eq:relation_Chois}\\
      &  \left(\mM(\gamma)^{-1/2}\otimes \sqrt{\gamma^{T}}\right) C_{\mM} \left(\mM(\gamma)^{-1/2}\otimes \sqrt{\gamma^{T}}\right)\;,\nonumber
    \end{align}
    where the superscript ${\bullet}^T$ denotes the matrix transposition done with respect to the fixed bases used in Eqs.~\eqref{eq:direct-choi} and~\eqref{eq:reverse-choi}.
\end{lemma}

We now want to construct two objects, $Q_F$ and $Q_R^\gamma$, which, analogously to the joint distributions $P_F$ and $P^\gamma_R$, are able to capture \textit{both} the input \textit{and} output of the forward and reverse processes. Specifically, the marginals of the operator $Q_F$ should recover the input state $\rho$ and the output $\mM(\rho)$ respectively, and analogously for $Q_R^\gamma$.

One choice is to define
\ba \label{eq:Q_F}
    Q_F := (\openone_B \otimes \sqrt{\rho^T}) C_{\mM} (\openone_B \otimes \sqrt{\rho^T})\,.
\ea
Such an operator is indeed able to capture the input and output of the forward process, in the sense that:
\ba \label{eq:TrQFQFdag}
    \Tr_A[Q_F] = \mM(\rho), \quad \Tr_B[Q_F] = \rho^T\,.
\ea

We define the representation for the reverse process $\mMR^\gamma$ similarly as
\begin{align} \label{eq:Q_R}
    Q_R^{\gamma}
    &:= \left(\sqrt{\tau}\otimes \openone_A \right) C_{\mMR^\gamma}^T  \left(\sqrt{\tau}\otimes \openone_A \right) \,,
\end{align}
where $\tau = \mM(\rho)$ is the input of the reverse process.
We use the transpose of the Choi operator of the reverse process so that it can be linked to $C_{\mM}$ by Lemma \ref{lem:choi} in the following way:
\begin{align*}
    &Q_R^{\gamma}=\\
    &\left(\sqrt{\tau}\sqrt{\mM{(\gamma)}}^{-1} \otimes \sqrt{\gamma^T}\right) C_{\mM} \left(\sqrt{\mM{(\gamma)}}^{-1} \sqrt{\tau}\otimes \sqrt{\gamma^T}\right) \,.
\end{align*}
This operator captures the input and output of the reverse process:
\ba \label{eq:TrQRQRdag}
    \Tr_A[Q_R^\gamma] = \tau, \quad \Tr_B[Q_R^\gamma] = [\mMR^\gamma(\tau)]^T\,.
\ea
The operators $Q_F,Q_R^\gamma$ just defined are analogous to  the state over time proposed by Leifer and Spekkens \cite{leifer2007conditional,leifer2013towards} up to a partial transpose.

Other definitions of input-output operators may satisfy nice properties. An alternative choice is, for example,
\begin{align}\label{eq:t-QF}
    {}^tQ_F := \sqrt{C_\mM}(\openone_B\otimes\rho^T)\sqrt{C_\mM}
\end{align}
and
\begin{align}\label{eq:t-QR}
    &{}^tQ_R^\gamma :=\\
    &\sqrt{C_\mM}(\mM(\gamma)^{-1/2}\tau\mM(\gamma)^{-1/2}\otimes\gamma^T)\sqrt{C_\mM}\;. \nonumber
\end{align}
The superscript ${}^t\bullet$ in~\eqref{eq:t-QF} and~\eqref{eq:t-QR} (not to be confused with $\bullet^T$) is used because the operators ${}^tQ_F$ and ${}^tQ_R^\gamma$ are, in a loose sense, a ``transposition'' of $Q_F$ and $Q_F^\gamma$, respectively. If $\Pi_y,\rho,\gamma$ do not commute, in general ${}^tQ_F\neq Q_F$ and ${}^tQ_R^\gamma\neq Q_R^\gamma$. For example, $\Tr_B[{}^tQ_F] = (\sum_y \sqrt{\Pi_y}\rho\sqrt{\Pi_y})^T$ which in general differs from $\rho^T$. Yet, they are similar, in the sense that $Q_F$ and ${}^tQ_F$ (resp. $Q_R^\gamma$ and ${}^tQ_R^\gamma$) share the same eigenvalues, and are thus unitarily equivalent, as it happens when doing a proper transposition.
Therefore, ${}^tQ_F$ and ${}^tQ_R^\gamma$ can be viewed as legitimate representations (up to unitaries) of the forward and reverse processes, and they will be useful in the irretrodictability interpretation of OE.

\subsection{Candidates for generalized OE}

Eqs.~\eqref{eq:sigma_rho_gamma_initial} and~\eqref{eq:Sj} provide two forms of the observational entropy: Eq.~\eqref{eq:sigma_rho_gamma_initial}, arising from the statistical deficiency approach, is the difference between relative entropies evaluated on the input system and the output system; Eq.~\eqref{eq:Sj}, arising from the irretrodictability approach, is the relative entropy between the forward and reverse processes. In the remainder of this section, we will propose generalizations of OE that take either or both of these forms.

\subsubsection{Candidate \#\numU: difference between input/output Umegaki entropies}

A first fully quantum generalisation of OE may just be obtained by replacing the reference state $u$ in Eq.~\eqref{SigmaBSS} with a general reference state $\gamma$:
\begin{align}
    &S^{(1)}_{\mathsf{M},\gamma}(\rho) \coloneqq\, S(\rho)+\Sigma^{(1)}_{\mathsf{M},\gamma}(\rho)  \label{SM1} \\&~\textrm{with}\;\Sigma^{(1)}_{\mathsf{M},\gamma}(\rho) = D(\rho\|\gamma)-D(\mathcal{M}(\rho)\|\mathcal{M}(\gamma))\nonumber\,,
\end{align}
that is, $S^{(1)}_{\mathsf{M},\gamma}(\rho) = -\Tr[\rho\ \ln\gamma]-D(\mathcal{M}(\rho)\|\mathcal{M}(\gamma))$, cf. Eq.~\eqref{eq:clax-gen-OE}, though this time it may be that $[\rho,\gamma]\neq 0$.
This definition has the form of Eq.~\eqref{eq:sigma_rho_gamma_initial} with $\mathfrak{D}$ taken to be the Umegaki relative entropy \eqref{eq:D_Umegaki}.
Notice that, while $D(\rho\|\gamma)$ is a fully quantum relative entropy, $D(\mathcal{M}(\rho)\|\mathcal{M}(\gamma))$ is in fact classical, since all the outputs of the channel $\mathcal{M}$ are diagonal in the same basis. %

\subsubsection{Candidate \#\numDQ: Umegaki relative entropy between forward/reverse processes}

Another option is to define a OE through Eq.~(\ref{eq:Sj}), thus choosing to prioritize irretrodictability. For this, one needs to choose a relative entropy and representations of the forward and reverse processes. Using the Umegaki relative entropy and the representations defined in \eqref{eq:Q_F} and \eqref{eq:Q_R}, we get
\begin{align}\label{eq:Sigma_DQ}
    \SDQ(\rho)\coloneqq\,&S(\rho)+\SigmaDQ(\rho)\\ &\textrm{with}\;\SigmaDQ(\rho) = D(Q_F\|Q_R^\gamma) .\nonumber
\end{align}
However, we will show in the following sections that this candidate lacks some of the properties we desire: we introduced it mainly for comparison with other candidates.

\subsubsection{Candidate \#\numBS: Belavkin--Staszewski relative entropy}

Besides Umegaki relative entropy, there are other choices for the quantum relative entropy between the representations of the forward and reverse processes, and between the states $\rho$ and $\gamma$. One such choice is the  Belavkin--Staszewski relative entropy \cite{belavkin1982c}, defined as
\ba\DBS(\rho\|\sigma)\coloneqq\Tr[\rho\ \ln\rho\sigma^{-1}] . \label{eq:DBS}\ea
The Belavkin--Staszewski relative entropy coincides with the Umegaki relative entropy and the classical relative entropy when $\rho$ and $\sigma$ commute, otherwise in general it is never smaller than Umegaki's. For a summary of the main properties of Belavkin--Staszewski relative entropy, and its relations with other quantum relative entropies, we refer the interested reader to Ref.~\cite{khatri2020principles}.

Inserting $\DBS$ into Eq.~\eqref{eq:sigma_rho_gamma_initial}, we obtain
\begin{align}
    &\SBS(\rho) \coloneqq S(\rho)+\SigmaBS(\rho)\label{SM2}\\ &~\textrm{with}\;\SigmaBS(\rho) = \DBS(\rho\|\gamma)-D(\mathcal{M}(\rho)\|\mathcal{M}(\gamma)) \nonumber\,.
\end{align}
Remarkably, it turns out that the above definition recovers the form of Eq.~\eqref{eq:Sj}. Assuming that $Q_R^\gamma$, and thus ${}^tQ_R^\gamma$, is full-rank, one has
\begin{align} \label{eq:DBS-D=DBS}
    \DBS(\rho \| \gamma)  - \DBS(\mathcal{M}(\rho) \| \mathcal{M}(\gamma)) = \DBS \left({}^tQ_F \middle\| {}^tQ_R^\gamma\right),
\end{align}
where $\DBS(\mathcal{M}(\rho) \| \mathcal{M}(\gamma))=\D(\mathcal{M}(\rho) \| \mathcal{M}(\gamma))$ since those states commute, and where ${}^tQ_F$ and ${}^tQ_R^\gamma$ were defined in~\eqref{eq:t-QF} and~\eqref{eq:t-QR}. The proof of the identity~\eqref{eq:DBS-D=DBS} is given in~\cref{app:proof_DDD}. Thus, $\SigmaBS$ indeed admits both the statistical deficiency and the irretrodictability interpretations.

\section{Properties}\label{sec:properties}

\begin{table*}
    \centering
    \scalebox{0.94}{
    \small
    \begin{tabular}{c|c|c|c|c|c}
    \hline
         \begin{minipage}{1.5cm}
         \centering
            Definition
         \end{minipage} & \begin{minipage}{4.2cm}
         \centering
            Deficiency interpretation
         \end{minipage} & \begin{minipage}{2.7cm}
         \centering
            Irretrodictability interpretation
         \end{minipage} & \begin{minipage}{2.5cm}
         \centering
            Equal to $\Scl(\rho)$ when
         \end{minipage} & \begin{minipage}{1.5cm}
         \centering
            Petz recovery criterion %
         \end{minipage} & \begin{minipage}{2.5cm}
         \centering
            \vspace{0.2em}
            Non-decreasing under stochastic post-processing
            \vspace{0.2em}
         \end{minipage}\\
         \hline
         $\SUmegaki$ \eqref{SM1} & $D(\rho\|\gamma) - D(\mM(\rho)\|\mM(\gamma))$ & N/A & $[\rho,\gamma]=0$
         & Yes & Yes\\
         \hline
         $\SDQ$ \eqref{eq:Sigma_DQ} & N/A & $D(Q_F\|Q_R^\gamma)$ & $\rho,\gamma,\Pi_y$ commute
         & Yes
           & No\\
         \hline
         $\SBS$ \eqref{SM2} & $\DBS(\rho\|\gamma) - D(\mM(\rho)\|\mM(\gamma))$ & $\DBS({}^tQ_F\|{}^tQ_R^\gamma)$ & $[\rho,\gamma]=0$
         & Yes & Yes\\
         \hline
    \end{tabular}
    }
    \caption{\label{tab:prop}Properties of $S^{(j)}_{\mathsf{M},\gamma}$. %
    The expressions for the statistical deficiency and irretrodictability interpretations do not match if one uses the Umegaki relative entropy. On the other hand, the use of the Belavkin-Staszewski relative entropy gives an expression that unifies both interpretations.}
\end{table*}

We proceed now to discuss the properties of the three candidates $S^{(j)}_{\mathsf{M},\gamma}$ ($j=1,2,3$) defined above, with a comparison between them summarized in \cref{tab:prop}. The main properties to consider for any candidate generalized OE are the following:
\begin{enumerate}[label={(\roman*)}]
    \item \label{item:first_essential_property} \label{item:property_u} When the reference prior is the uniform distribution (maximally mixed state), the candidate should recover the original definition \eqref{eq:original}. This is true for $S^{(1)}_{\mathsf{M},\gamma}$ and $S^{(3)}_{\mathsf{M},\gamma}$, namely when $\gamma=u:=\openone/d$,
    \begin{align}
        S_{\mathsf{M},u}^{(1,3)}(\rho) = S_{\mathsf{M}}(\rho) \,.
    \end{align}
    Instead, in order to recover the conventional OE, $\SDQ$ further requires that $[\rho,\Pi_y]=0$ for all $y$.
    
    \item \label{item:property_clax} More generally, when $[\rho,\gamma]=0$, one has \begin{align}
        S_{\mathsf{M},\gamma}^{(1,3)}(\rho)=\Scl(\rho) \,.
    \end{align}
    Instead, the condition $S_{\mathsf{M},\gamma}^{(2)}(\rho)=\Scl(\rho)$ in general requires $[\rho,\gamma]=[\rho,\Pi_y]=[\gamma,\Pi_y]=0$ for all $y$.
    
    \item \label{item:property_lower_bound} Like the original OE, all of them are lower-bounded by the von Neumann entropy:
    \begin{align}
        S^{(j)}_{\mathsf{M},\gamma}(\rho)\geq S(\rho)\,.
    \end{align} %
    Thus, the OEs retain the desirable property that one cannot have less uncertainty than the von Neumann entropy. %
    \label{item:last_essential_property}
\end{enumerate}
The proofs of the above properties are in \cref{app:essential_properties}.

Other non-essential, yet desirable properties include:

\begin{enumerate}[resume,label={(\roman*)}]
    \item \label{item:property_diff}
    $S_{\mathsf{M},\gamma}^{(j)}(\rho)$ admits both interpretations, as statistical deficiency, i.e. Eq.~\eqref{eq:sigma_rho_gamma_initial}, and irretrodictability, i.e., Eq.~\eqref{eq:Sj}. This property is satisfied by $\SBS$, with suitable definitions of the input-output descriptions.

    \item \label{item:property_Petz} $S^{(j)}_{\mathsf{M},\gamma}(\rho)$ satisfies the Petz recovery criterion: $S^{(j)}_{\mathsf{M},\gamma}(\rho)=S(\rho)$ if and only if $\mMR^\gamma \big(\mM(\rho)\big)=\rho$, where $\mMR^\gamma$ is the Petz map of $\mM$ with reference $\gamma$ defined in \eqref{eq:reverse_process}. This property is satisfied by all candidates, as shown in \cref{app:Petz}. %

    \item \label{item:property_monotonicity} $S^{(j)}_{\mathsf{M},\gamma}(\rho)$ is non-decreasing under stochastic post-processing. We say $\mathsf{M}'=\{\Pi'_z\}$ is a post-processing of $\mathsf{M}$ if its outcome can be obtained by applying a stochastic map on the outcome of $\mathsf{M}$, namely there exists a stochastic matrix $w$ with $\sum_z w_{zy}=1$ for all $y$ satisfying
    \begin{align}
        \Pi'_z = \sum_i w_{zy} \Pi_y, ~\forall y \,.
    \end{align}
    This property for $S^{(j)}_{\mathsf{M},\gamma}(\rho)$ says that, for any $\mathsf{M}'$ that is a post-processing of $\mathsf{M}$, one has
    \begin{align}
        S^{(j)}_{\mathsf{M}',\gamma}(\rho) \geq S^{(j)}_{\mathsf{M},\gamma}(\rho)\,.
    \end{align}
    This property is satisfied by $j=\numU,\numBS$, with proofs in \cref{app:stochastic}.
    
\end{enumerate}

Finally, we notice that while the original OE, Eq.~\eqref{eq:original}, is upper bounded as $S_{\mathsf{M}}(\rho) \leq \ln d$, in general, for a non-uniform reference $\gamma$, the same bound does not hold, as expected. However 
\begin{align}
    \SUmegaki(\rho) &\leq \SBS(\rho)\label{eq:Ad_between_U_BS}
\end{align}
holds because the Belavkin-Staszewski relative entropy bounds the Umegaki one from above \cite{matsumoto2015new,hiai2017different}. Also
\begin{align}
    \SUmegaki(\rho) &\leq \SDQ(\rho)\label{eq:ineq_12} 
\end{align}
holds due to joint convexity of the relative entropy (proof in \cref{app:proof_ineq_12}).

\section{Examples}\label{sec:examples}

\subsection{Gibbs prior}

In the presence of a Hamiltonian $H = \sum_{n=0}^{d-1} E_n\ketbra{n}$, a very natural choice of non-uniform prior is the Gibbs state
\begin{align}
    \gamma := e^{-\beta H}/\Tr[e^{-\beta H}], ~ \beta>0.
\end{align}

We also consider the measurement in the energy eigenbasis $\mathsf{M}:= \{\ketbra{0},\dots,\ketbra{d-1}\}$, but to move far away from the classical case we assume that the input state is pure and maximally unbiased with the energy eigenbasis:
\begin{align}
    \rho := \ketbra{\psi}&\textrm{ with } \ket{\psi} := \frac{1}{\sqrt{d}}\sum_{n=0}^{d-1} \ket{n}.
\end{align}
With these assumptions, the first definition yields
\begin{align}
    \SUmegaki(\rho) &= S_{\mathsf{M}}(\rho) = \ln d\,,
\end{align} which is also the case if $\rho$ is a mixture of maximally unbiased states. Like $\Scl$, $\SUmegaki$ reduces to the original $S_{\mathsf{M}}$ when the prior is a convex sum of the measurement elements.

The second definition yields
\begin{align}
    \SDQ(\rho) &= \infty\,,
\end{align} for any pure state, since the support of $Q_R^\gamma=\frac1d\sum_n\ketbra{n}\otimes\ketbra{n}$ does not contain the support of $Q_F=\frac1d \openone\otimes\ketbra{\psi}$. We shall comment on this result after the next example. 

Finally, the third definition yields
\begin{align}
    \SBS(\rho) &= \ln\Tr[\gamma^{-1}] + \frac{1}{d}\Tr[\ln\gamma] \\
    &= \ln\frac{1-e^{\beta\omega d}}{1-e^{\beta\omega}} - \frac{\beta\omega(d-1)}{2}\,.
\end{align}
where the first line is general, while the second is the expression for equidistant spectrum $E_n=n\omega$. Thus $\SBS$ is more sensitive than $\SUmegaki$ to quantum situations.

\subsection{Three-qubit encoding}

The following example is inspired by a simple error-correcting code, the three-qubit encoding of a pure qubit:
\begin{align}
    \alpha \ket0 + \beta \ket1 \mapsto \alpha \ket{000} + \beta \ket{111}, ~ |\alpha|^2+|\beta|^2 =1 \,.
\end{align}
Suppose $\rho$ is the encoded state
\begin{align}
    \rho := (\alpha \ket{000} + \beta \ket{111})(\alpha^* \bra{000} + \beta^* \bra{111}) \,.
\end{align}
We consider the measurement of each qubit in the $\{\ket{+},\ket{-}\}$ basis, i.e.~the POVM elements are projectors on the basis vectors
\begin{align}
    \{\ket{+++},\ket{++-},\ket{+-+},\dots,\ket{---}|\}\,.
\end{align}
As for the prior, we suppose that the observer knows the encoding of the error correction code, and expect $\rho$ to be more probably in the subspace spanned by $\ket{000}$ and $\ket{111}$, possibly with a bias towards one of those product states; whence
\begin{align}
    \gamma &:= p_0\ketbra{000} + p_1\ketbra{111} \nonumber \\
    &~+ \frac{1-p_0-p_1}{6}(\openone - \ketbra{000} - \ketbra{111}) \,, \\
    & p_0,p_1>0,~ p_0+p_1<1\,. \nonumber
\end{align}
In this case, the three definitions proposed here yield
\begin{align}
    \SUmegaki(\rho) &= |\alpha|^2\ln\frac{1}{p_0} + |\beta|^2\ln\frac{1}{p_1} - D(\mM(\rho)\|\mM(\gamma))\,,\\
    \SDQ(\rho) &= \infty\,,\\
    \SBS(\rho) &= \ln\left(\frac{|\alpha|^2}{p_0} + \frac{|\beta|^2}{p_1}\right) - D(\mM(\rho)\|\mM(\gamma))
\end{align} with
\begin{align}
    &D(\mM(\rho)\|\mM(\gamma))\nonumber\\&= |\alpha + \beta|^2\ln|\alpha+\beta|+|\alpha - \beta|^2\ln|\alpha-\beta|\,.
\end{align}
$\SUmegaki$ and $\SBS$ differ in the first term, as long as $p_0\neq p_1$: for $p_0=p_1= p$, both yield $\ln(1/p)$. In particular, when $p=\frac{1}{8}$, $\gamma=u$ and therefore $\SUmegaki(\rho)=\SBS(\rho)=S_{\mathsf{M}}(\rho)$.  
We see that $\SDQ$ is still infinite, for the same reason of support mismatch as in the previous example. From the examples, we observe that $\SDQ$ is often overly sensitive to the non-commutativity between $\rho$ and $\gamma$. This suggests that, instead of the natural choice of $Q_F$ and $Q_R^\gamma$ as input-output representations, one could opt for representations whose supports are more aligned, such as ${}^tQ_F$ and ${}^tQ_R^\gamma$, which relate to $\SBS$ via Eq.~\eqref{eq:DBS-D=DBS}.

\section{Conclusions}

The original definition [Eq.~\eqref{eq:original}] of observational entropy (OE) was known to be lower-bounded by the von Neumann entropy. Here we have first brought to the fore that the excess term $\Sigma_{\mathsf{M}}(\rho)$ can be interpreted in two ways: as a statistical deficiency \eqref{SigmaBSS}, quantifying the decrease of state distinguishability induced by the measurement; and as irretrodictability \eqref{eq:classical_D=S},  quantifying the hardness of retrodicting the input from the output statistics. While it is intuitive that recovering the input state is harder if the measurement makes states less distinguishable, the exact coincidence of the quantifiers is of interest.

In both interpretations, we observe that the uniform state $u$ plays the role of reference, or prior, knowledge. This may not represent the proper knowledge of the physical situation: for instance, for systems in contact with a thermal bath, it may be more natural to choose the Gibbs prior. Based on this, we have studied generalisations of OE, in which the prior knowledge can be an arbitrary state $\gamma$.

When $[\rho,\gamma]=0$, we find an obvious generalisation of the excess term [Eq.~\eqref{eq:D-D=DP}] that retains both interpretations of statistical deficiency and irretrodictability. This is no longer straightforward for a general quantum prior. Technically, one of the main difficulty lies in that the irretrodictability quantifier is a relative entropy between joint input-and-output objects, whose definition in quantum theory is a current topic of research. We have explored three possible definitions of generalized OE (Table~\ref{tab:prop}): two specifically designed to satisfy one of the interpretations but lacking the other; the third retaining both by replacing the usual Umegaki relative entropy with the Belavkin-Staszewski version. Thus we have a novel fully quantum object, that quantifies simultaneously the loss of distinguishability by the measurement and the hardness to retrodict the input knowing the output. Being built from information-theoretical considerations,  our new formulation of OE may also hold significance in physical (thermodynamical) contexts, such as its relationship with work extraction \cite{safranek2023work}. We leave a deeper exploration of the physical implications of OE for future research.

\medskip

\section*{Acknowledgments}

We thank Clive Aw, Fumio Hiai, Anna Jen\v{c}ov\'a and Andreas Winter for discussions.

G.B.~and V.S.~are supported by the National Research Foundation, Singapore and A*STAR under its CQT Bridging Grant; and by the Ministry of Education, Singapore, under the Tier 2 grant ``Bayesian approach to irreversibility'' (Grant No.~MOE-T2EP50123-0002).  D.\v{S}.~acknowledges the support from the Institute for Basic Science in Korea (IBS-R024-D1). F.B.~acknowledges support from MEXT Quantum Leap Flagship Program (MEXT QLEAP) Grant No.~JPMXS0120319794, from MEXT-JSPS Grant-in-Aid for Transformative Research Areas (A) ``Extreme Universe'' No.~21H05183, and from JSPS KAKENHI, Grants No.~20K03746 and No.~23K03230. J.S.~acknowledges support by MICIIN with funding from European Union NextGenerationEU (PRTR-C17.I1) and by Generalitat de Catalunya.

\bibliographystyle{quantum}
\bibliography{refsretro}

\newpage\onecolumn
\appendix

\section{Proof of Lemma \ref{lem:choi}} \label{app:proof_choi}
\cref{lem:choi} relates the Choi operators of the forward and reverse processes. This can be shown using their definitions \eqref{eq:direct-choi} and \eqref{eq:reverse-choi}.
\begin{proof}
    Let the Kraus representation of $\mM$ be $\mM(\rho) = \sum_k K_k \rho K_k^\dag$. 
    First, observe the following identity
    \ba
        \sum_{i,j} A\ket{i}\bra{j}B \otimes \ket{i}\bra{j} = \sum_{i,j} \ket{i}\bra{j} \otimes A^T\ket{i}\bra{j}B^T
    \ea
    for any operators $A$ and $B$. 
    Using this identity twice, the right-hand side of \cref{eq:relation_Chois} equals to
    \begin{align}
        &\phantom{=}~ \left(\mM(\gamma)^{-1/2}\otimes \sqrt{\gamma^{T}}\right) C_{\mM} \left(\mM(\gamma)^{-1/2}\otimes \sqrt{\gamma^{T}}\right) \nonumber \\
       & = \sum_{i,j} \mM(\gamma)^{-1/2}\mM(\ket{i}\bra{j})\mM(\gamma)^{-1/2} \otimes \sqrt{\gamma^{T}} \ket{i}\bra{j}\sqrt{\gamma^{T}} \nonumber\\
        & = \sum_{i,j,k} \mM(\gamma)^{-1/2}\ket{i}\bra{j}\mM(\gamma)^{-1/2} \otimes \sqrt{\gamma^{T}} K_k^T \ket{i}\bra{j} K_k^*\sqrt{\gamma^{T}} \nonumber\\
        & = \sum_{i,j,k} \ket{i}\bra{j} \otimes  \sqrt{\gamma^{T}} K_k^T (\mM(\gamma)^{-1/2})^T\ket{i}\bra{j}(\mM(\gamma)^{-1/2})^T K_k^*\sqrt{\gamma^{T}} \label{eq:CmM_expansion}
    \end{align}
    On the other hand, 
    \begin{align}
        C_{\mMR^\gamma} &= \sum_{i,j} \ket{i}\bra{j} \otimes \mMR^\gamma (\ket{i}\bra{j}) \nonumber\\
        &=\sum_{i,j,k} \ket{i}\bra{j} \otimes \sqrt{\gamma} K_k^\dag \mM(\gamma)^{-1/2}\ket{i}\bra{j}\mM(\gamma)^{-1/2} K_k \sqrt{\gamma}\nonumber \\
        &=\sum_{i,j,k} \ket{j}\bra{i} \otimes \sqrt{\gamma} K_k^\dag \mM(\gamma)^{-1/2}\ket{j}\bra{i}\mM(\gamma)^{-1/2} K_k \sqrt{\gamma} \label{eq:CmMR_expansion}
    \end{align}
    Notice that \cref{eq:CmMR_expansion} is the transpose of \cref{eq:CmM_expansion}. This proves \cref{eq:relation_Chois}.
\end{proof}

\section{Proof of \cref{eq:DBS-D=DBS}} \label{app:proof_DDD}

The most important observation for \cref{eq:DBS-D=DBS} is that, one of the $\sqrt{C_{\mM}}$ in the definitions of ${}^tQ_F$ (\ref{eq:t-QF}) and ${}^t Q_R$ (\ref{eq:t-QR}) will cancel each other in the expression of  $\DBS\left({}^tQ_F \middle\| {}^tQ_R^{\gamma}\right) $, leaving a tensor product inside the logarithm. That is to say,
    \begin{align}
        &\phantom{=}\ \ln {}^tQ_F ({}^tQ_R^{\gamma})^{-1} \nonumber\\
        & = \ln \sqrt{C_{\mathcal{M}}} (\openone \otimes \rho^T) \sqrt{C_{\mathcal{M}}} \left(\sqrt{C_{\mathcal{M}}} (\mathcal{M}(\gamma)^{-1/2} \tau \mathcal{M}(\gamma)^{-1/2} \otimes \gamma^T)\sqrt{C_{\mathcal{M}}}\right)^{-1} \nonumber\\
        & = \ln \sqrt{C_{\mathcal{M}}} ( \mathcal{M}(\gamma)^{1/2} \tau^{-1} \mathcal{M}(\gamma)^{1/2} \otimes \rho^T (\gamma^T)^{-1}) \sqrt{C_{\mathcal{M}}}^{-1} \nonumber\\
        & =\sqrt{C_{\mathcal{M}}} ( \ln \mathcal{M}(\gamma)^{1/2} \tau^{-1} \mathcal{M}(\gamma)^{1/2} \otimes \rho^T (\gamma^T)^{-1}) \sqrt{C_{\mathcal{M}}}^{-1} \nonumber\\
        & =\sqrt{C_{\mathcal{M}}} \left( \ln \mathcal{M}(\gamma)^{1/2} \tau^{-1} \mathcal{M}(\gamma)^{1/2} \otimes \openone + \openone \otimes \ln \rho^T (\gamma^T)^{-1} \right) \sqrt{C_{\mathcal{M}}}^{-1}  \label{eq:tQFtQR_sep}
    \end{align}
    
Notice that ${}^tQ_R^\gamma$ being full-rank implies $C_{\mM}$ and $\gamma$ being full-rank. Putting \cref{eq:tQFtQR_sep} into the definition of $\DBS$, the left-hand side of \cref{eq:DBS-D=DBS} is
     \begin{align}
        &\phantom{=}\ \DBS\left({}^tQ_F \middle\| {}^tQ_R^{\gamma}\right) \nonumber\\
        &= \Tr[{}^tQ_F\ \ln {}^tQ_F ({}^tQ_R^{\gamma})^{-1} ] \nonumber \\
        &= \Tr\left[ (\openone \otimes \rho^T) C_{\mathcal{M}} \left( \ln \mathcal{M}(\gamma)^{1/2} \tau^{-1} \mathcal{M}(\gamma)^{1/2} \otimes \openone + \openone \otimes \ln \rho^T (\gamma^T)^{-1} \right)   \right] \nonumber \\
        &= \Tr[ \Tr_B[ (\openone \otimes \rho^T) C_{\mathcal{M}} ]\ \ln \rho^T (\gamma^T)^{-1} ] + \Tr[ \Tr_A[ (\openone \otimes \rho^T) C_{\mathcal{M}} ]\ \ln \mathcal{M}(\gamma)^{1/2} \tau^{-1} \mathcal{M}(\gamma)^{1/2}  ] \nonumber \\
        &= \Tr[ \rho^T \ \ln \rho^T(\gamma^T)^{-1} ] + \Tr[ \mathcal{M}(\rho)\ \ln \mathcal{M}(\gamma)^{1/2} \tau^{-1} \mathcal{M}(\gamma)^{1/2}  ] \nonumber \\
        &=\DBS(\rho\|\gamma) - \Tr[\mathcal{M}(\rho) \ \ln \mathcal{M}(\gamma)^{-1/2}\tau\mathcal{M}(\gamma)^{-1/2}]
    \end{align}
    Recall that we choose the input of the reverse process to be $\tau = \mM(\rho)$. Noting that $\mM(\rho)$ commutes with $\mM(\gamma)$, the second term equals to
    \begin{align}
        & \phantom{=} \ \Tr[\mathcal{M}(\rho) \ \ln \mathcal{M}(\gamma)^{-1/2}\tau\mathcal{M}(\gamma)^{-1/2}] \nonumber\\
        &=\Tr[\mM(\rho) (\ln \mM(\rho) - \ln\mM(\gamma))] \nonumber\\
        &= D(\mM(\rho)\|\mM(\gamma)) \,.
    \end{align}
    This proves \cref{eq:DBS-D=DBS}.

\section{Proof of properties \ref{item:first_essential_property}-\ref{item:last_essential_property}} \label{app:essential_properties}

Since $S_{\mathsf{M},u}^{\rm clax}(\rho) = S_{\mathsf{M}}(\rho)$ and $[\rho,u]=0$, property \ref{item:property_u} is a special case of property \ref{item:property_clax}, so we prove \ref{item:property_clax} directly.

When $[\rho,\gamma]=0$, $D(\rho\|\gamma)$ and $\DBS(\rho\|\gamma)$ are both equal to the relative entropy between the eigenvalues of $\rho$ and $\gamma$. By comparing their definitions \cref{Scl1,SM1,SM2}, we obtain $\SUmegaki(\rho) = \SBS(\rho) = \Scl(\rho)$.

For $\SDQ$, we further need to use that $[\rho,\Pi_y]=[\gamma,\Pi_y] = 0$. This condition indicates that $C_{\mM},(\openone\otimes\rho),(\openone\otimes\gamma)$ all commute, and therefore $Q_F = (\openone\otimes\rho^T)C_{\mM}$, $Q_R^\gamma=(\mM(\rho)\mM(\gamma)^{-1} \otimes \gamma^T) C_{\mM}$, and
\begin{align}
    \SigmaDQ(\rho)  & = D(Q_F \| Q_R^\gamma)  \nonumber \\
    & = \Tr[Q_F (\ln Q_F - \ln Q_R^\gamma)] \nonumber\\
    & = \Tr[Q_F (\ln (\openone\otimes\rho^T) + \ln C_{\mM})]  - \Tr\left[Q_F\left(\ln (\mM(\rho)\mM(\gamma)^{-1} \otimes \gamma^T) + \ln C_{\mM} \right)\right] \nonumber \\
    & = \Tr\left[Q_F \left(\openone \otimes (\ln \rho^T - \ln \gamma^T)\right)\right]  -\Tr[Q_F (\ln \mM(\rho)\mM(\gamma)^{-1} \otimes \openone)] \nonumber \\
    & =  D(\rho\|\gamma)-D(\mM(\rho)\|\mM(\gamma)) \nonumber\\
    & = \Scl(\rho) - S(\rho)
\end{align}
Notice that the above equality holds as long as the support of $Q_F$ is contained in that of $Q_R^\gamma$ without assuming $Q_R^\gamma$ to be invertible, since the Umegaki relative entropy is continuous with respect to both arguments \cite{bluhm2023continuity}.
Therefore, $\SDQ(\rho) = \Scl(\rho)$ and  property \ref{item:property_clax} holds for $j=\numDQ$.

Property \ref{item:property_lower_bound} is equivalent to say that $\Sigma_{\mathsf{M},\gamma}^{(j)}$ is non-negative.

$\SigmaDQ(\rho)=D(Q_F\|Q_R^\gamma)$ is non-negative by the non-negativity of relative entropy between two unit-trace positive operators.

$\SigmaU(\rho)=D(\rho\|\gamma) - D(\mM(\rho)\|\mM(\gamma))$ is non-negative by the data-processing inequality of Umegaki relative entropy. Last, by \cref{eq:Ad_between_U_BS}, $\SigmaBS(\rho)  \geq \SigmaU(\rho) \geq 0$.

\section{Proof of Petz recovery criteria \ref{item:property_Petz}}\label{app:Petz}

We first show that both $\SUmegaki$ and $\SBS$ are equal to $S(\rho)$ if and only if $\mMR^\gamma(\mM(\rho)) =\rho$. 
The $j=\numDQ$ case is addressed later. 

\subsection{Property \ref{item:property_Petz} of $\SUmegaki$}
The property $D(\rho\|\gamma) = D(\mM(\rho)\|\mM(\gamma)) \Leftrightarrow \mMR^\gamma(\mM(\rho)) =\rho$  is shown, for the larger family called $f$-divergences, by Theorem 5.1 in Ref. \cite{hiai2011quantum} and Theorem 3.18 in Ref. \cite{hiai2017different}. Therefore, the property \ref{item:property_Petz} for $\SUmegaki$ is proved since $\SUmegaki(\rho)=S(\rho)$ is equivalent to $D(\rho\|\gamma) = D(\mM(\rho)\|\mM(\gamma))$.

\subsection{Property \ref{item:property_Petz} of $\SBS$}
For $\SBS$, first assume $\mMR^\gamma(\mM(\rho)) =\rho$. By the data-processing inequality of the Belavkin-Staszewski relative entropy \cite{hiai2017different,bluhm2020strengthened},
\begin{align}
    \DBS(\rho\|\gamma) &\geq \DBS(\mM(\rho)\|\mM(\gamma))  \nonumber\\
    &\geq \DBS\left(\mMR^\gamma(\mM(\rho)) \middle\| \mMR^\gamma(\mM(\gamma) ) \right) \nonumber\\
    &= \DBS(\rho\|\gamma) \,,
\end{align}
where $\mMR^\gamma(\mM(\gamma)) = \gamma$ by definition. 
Therefore, $\DBS(\rho\|\gamma) =\DBS(\mM(\rho)\|\mM(\gamma))$ and thus $\SBS(\rho) = S(\rho)$.

On the other hand, if $\SBS(\rho) = S(\rho)$, by \cref{eq:Ad_between_U_BS} and property \ref{item:property_lower_bound}, $\SBS(\rho)\geq\SUmegaki(\rho)\geq S(\rho)$, one infers $\SUmegaki(\rho)=S(\rho)$, and therefore $\mMR^\gamma(\mM(\rho)) =\rho$ by property \ref{item:property_Petz} of $\SUmegaki$.

\subsection{Property \ref{item:property_Petz} of $\SDQ$}

Next, we prove property \ref{item:property_Petz} for $\SDQ$. Before that, we notice that the Petz recovery condition $\mMR^\gamma(\mM(\rho)) = \rho$ indicates that $[\rho,\gamma]=0$.
\begin{lemma} \label{lem:Petz_commute}
    For a measurement  channel $\mM$, $\mMR^\gamma(\mM(\rho)) = \rho$ implies $[\rho,\gamma]=0$.
\end{lemma}
\begin{proof}
    By property \ref{item:property_Petz} of $\SUmegaki$ and $\SBS$, one has $\SigmaU(\rho) = \SigmaBS(\rho) = 0$. That is,
    \begin{align}
        D(\rho\|\gamma) - D(\mM(\rho)\|\mM(\gamma)) &= \DBS(\rho\|\gamma) - D(\mM(\rho)\|\mM(\gamma)) \nonumber\\
        D(\rho\|\gamma)  &= \DBS(\rho\|\gamma)
    \end{align}
    Because $D(\rho\|\gamma)  = \DBS(\rho\|\gamma)$ if and only if $[\rho,\gamma]=0$ (Ref. \cite{hiai2017different} Theorem 4.3), one has $[\rho,\gamma]=0$.
\end{proof}

Now, we prove property \ref{item:property_Petz}.
We first prove the ``only if'' part.
    
Suppose $S_{\mathsf{M},\gamma}^{(2)}=S(\rho)$. This is equivalent to that $\SigmaDQ(\rho) = D(Q_F\| Q_R^\gamma) = 0$, which is equivalent to that $Q_F = Q_R^\gamma$. Taking the partial trace over system $B$, we get $\rho^T = \Tr_B[Q_F] = \Tr_B[Q_R^\gamma] = \mMR^\gamma(\mM(\rho))^T$. Therefore,  $\mMR^\gamma(\mM(\rho)) = \rho$.

For the ``if'' part, we will show that $\mMR^\gamma(\mM(\rho)) = \rho$ implies $Q_F = Q_R^\gamma$, which is equivalent to $S_{\mathsf{M},\gamma}^{(2)}=S(\rho)$. 

Suppose $\mMR^\gamma(\mM(\rho)) = \rho$.
By \cref{lem:Petz_commute}, $[\rho,\gamma]=0$, and thus we can diagonalize them in the same basis:
\begin{align}
    \rho &= \sum_x \lambda_x \ketbra{\psi_x} \,, \\
    \gamma &= \sum_x \gamma_x \ketbra{\psi_x} \,.
\end{align}
Next, we construct a new POVM by taking the diagonal elements of the POVM $\mathsf{M}=\{\Pi_y\}$ in the above basis:
\begin{align}
    \mathsf{M}'\coloneqq\{\Pi_y'\},~~\Pi_y' \coloneqq \sum_x \ketbra{\psi_x} \Pi_y \ketbra{\psi_x} \,.
\end{align}
Let $\mM'(\rho) \coloneqq \sum_y \Tr[\Pi_y' \rho] \ketbra{y}$ be the measurement channel and $\mMR'^\gamma$ be its Petz map with reference $\gamma$.
Notice that $\Tr[\Pi_y\rho] = \Tr[\Pi_y'\rho]$, $\Tr[\Pi_y\gamma] = \Tr[\Pi_y'\gamma]$, and therefore $\mM'(\rho) = \mM(\rho)$, $\mM'(\gamma) = \mM(\gamma)$. 

Since $[\rho,\gamma] = [\rho,\Pi'_y] = [\gamma,\Pi'_y]=0$, using property \ref{item:property_clax} for $\mM'$, one has
\begin{align} 
    S_{\mathsf{M}',\gamma}^{\textrm{clax}}(\rho) -S(\rho) &= D(\rho\|\gamma) - D(\mM(\rho)\|\mM(\gamma)) \nonumber \\
    &= D(Q_{F'}\|Q_{R'}^\gamma) \,, \label{eq:clax_Mprime}
\end{align}
where $Q_{F'}\coloneqq (\openone \otimes \rho^T) C_{\mM'}$ and $Q_{R'}^\gamma \coloneqq (\mM'(\rho) \mM'(\gamma)^{-1} \otimes \gamma^T) C_{\mM'}$ are the representations of $\mM'$ and $\mMR'^\gamma$, which are simplified using the commutativity of $\rho$,$\gamma$ and $\Pi_y'$.

Note that the expression \cref{eq:clax_Mprime} is the same as $\SigmaU(\rho) = D(\rho\|\gamma) - D(\mM(\rho)\|\mM(\gamma))$. 
Since $\mMR^\gamma(\mM(\rho)) = \rho$, by property \ref{item:property_Petz} for $\SUmegaki$, we have  $\SigmaU(\rho) = 0$.
Combining this with \cref{eq:clax_Mprime} indicates that $D(Q_{F'}\|Q_{R'}^\gamma)=0$ and  $Q_{F'}=Q_{R'}^\gamma$. Expanding this with respect to the basis $\{ \ket{y} \otimes \ket{\psi_x} \}$, one get
\begin{align}
    \bra{y}\bra{\psi_x}Q_{F'}\ket{y}\ket{\psi_x}&=\bra{y}\bra{\psi_x}Q_{R'}^\gamma \ket{y}\ket{\psi_x} \\
    \varphi(y|x) \lambda_x& = \frac{\varphi(y|x)\gamma_x \Tr[\Pi'_y \rho]}{\Tr[\Pi'_y \gamma]}, ~\forall x, y  \label{eq:lambda_prop_gamma}
\end{align}
where $\varphi(y|x) = \bra{\psi_x}\Pi'_y\ket{\psi_x} = \bra{\psi_x}\Pi_y\ket{\psi_x}$.

Now, fix $y$, and consider the support of $\Pi'_y$, which is spanned by all $\ket{\psi_x}$ such that $\varphi(y|x)\neq 0$. In this subset of $x$, one could cancel out $\varphi(y|x)$ in both sides and \cref{eq:lambda_prop_gamma} becomes
\begin{align}
    \lambda_x = \frac{ \Tr[\Pi'_y \rho]}{\Tr[\Pi'_y \gamma]} \gamma_x, ~\text{for~} x:\ket{\psi_x}\in\text{Supp}(\Pi'_y) \,,
\end{align}
where $\text{Supp}(\Pi'_y)$ denotes the support of $\Pi'_y$.  Taking the square root of this equation, one gets
\begin{align}
    \sqrt{\lambda_x} = \sqrt{\frac{ \Tr[\Pi'_y \rho]}{\Tr[\Pi'_y \gamma]}}\sqrt{\gamma_x}, ~\text{for~} x:\ket{\psi_x}\in\text{Supp}(\Pi'_y) \,. \label{eq:sqrt_lambda_gamma}
\end{align}

Define $(\Pi_y')^0\coloneqq \sum_{x:\varphi(y|x)\neq 0}\ketbra{\psi_x}$ as the projector onto $\text{Supp}(\Pi'_y)$.
The above equation can be rewritten as
\begin{align} \label{eq:sqrt_rho_sqrt_gamma}
    \sqrt\rho (\Pi_y')^0   = \sqrt{\frac{ \Tr[\Pi'_y \rho]}{\Tr[\Pi'_y \gamma]}} \sqrt\gamma (\Pi_y')^0 \,,
\end{align}
since multiplication with $(\Pi_y')^0$ selects the eigenvectors of $\sqrt{\rho}$ and $\sqrt{\gamma}$  in $\text{Supp}(\Pi'_y)$, which are $\sqrt{\lambda_x}$ and $\sqrt{\gamma_x}$ in \cref{eq:sqrt_lambda_gamma}. 

Since $\Pi_y$ is positive, its off-diagonal elements are cross terms in $\text{Supp}(\Pi'_y)$:
\begin{align}
    \Pi_y  = \Pi'_y + \sum_{\substack{x_1,x_2:  x_1\neq x_2, \\\ket{\psi_{x_1}},\ket{\psi_{x_2}}\in\text{Supp}(\Pi_y')}} \pi_{y,x_1,x_2} \ket{\psi_{x_1}}\bra{\psi_{x_2}}
\end{align}
for some complex numbers $\pi_{y,x_1,x_2}$. Therefore, $\Pi_y (\Pi_y')^0 = (\Pi_y')^0\Pi_y = \Pi_y$.

By this and \cref{eq:sqrt_rho_sqrt_gamma},
\begin{align}\label{eq:srPisr}
    \sqrt{\rho} \Pi_y \sqrt{\rho} &= \sqrt{\rho} (\Pi_y')^0 \Pi_y (\Pi_y')^0  \sqrt{\rho} \nonumber\\
    & = \frac{ \Tr[\Pi'_y \rho]}{\Tr[\Pi'_y \gamma]}\sqrt{\gamma} (\Pi_y')^0 \Pi_y (\Pi_y')^0  \sqrt{\gamma} \nonumber \\
    & = \frac{ \Tr[\Pi'_y \rho]}{\Tr[\Pi'_y \gamma]}\sqrt{\gamma} \Pi_y \sqrt{\gamma} 
\end{align}

Last, by $\Tr[\Pi_y\rho] = \Tr[\Pi_y'\rho]$, $\Tr[\Pi_y\gamma] = \Tr[\Pi_y'\gamma]$ and \cref{eq:srPisr}, one gets
\begin{align}
    Q_F & = (\openone \otimes \sqrt{\rho^T}) C_{\mM} (\openone \otimes \sqrt{\rho^T}) \nonumber \\
    & = \sum_y \ketbra{y} \otimes \sqrt{\rho^T} \Pi_y^T \sqrt{\rho^T} \nonumber \\
    & = \sum_y \ketbra{y} \otimes\frac{ \Tr[\Pi'_y \rho]}{\Tr[\Pi'_y \gamma]}\sqrt{\gamma^T} \Pi_y^T \sqrt{\gamma^T} \nonumber \\
    & = \sum_y \ketbra{y} \otimes\frac{ \Tr[\Pi_y \rho]}{\Tr[\Pi_y \gamma]}\sqrt{\gamma^T} \Pi_y^T \sqrt{\gamma^T} \nonumber \\
    & = (\mM(\rho)\mM(\gamma)^{-1} \otimes \sqrt{\gamma^T}) C_\mM (\openone \otimes \sqrt{\gamma^T}) \nonumber \\
    & = Q_R^\gamma
\end{align}
Therefore, $Q_F = Q_R^\gamma$, and $S_{\mathsf{M},\gamma}^{(2)}=S(\rho)$.

\section{Proof of monotonicity under stochastic post-processing \ref{item:property_monotonicity}} \label{app:stochastic}

Let $\mathcal{W}$ be the linear map describing the post-processing $w$, which satisfies
\begin{align}
    \mathcal{W}(\ketbra{y}) &= \sum_z w_{zy} \ketbra{z} \\
    \mathcal{W}(\ket{y}\!\!\!\;\bra{y'})& = 0, \quad \text{for~} y\neq y' \,.
\end{align}

Since $w$ is a stochastic matrix, $\mathcal{W}$ is completely positive and trace-preserving. The measurement channel of $\mathsf{M}'$ is then described by
\begin{align}
    (\mathcal{P}\circ\mathcal{M}) (\rho) &= \sum_y \Tr[\Pi_y \rho] \mathcal{P}(\ketbra{y}) \nonumber \\
    &= \sum_z \Tr[\Pi_z' \rho] \ketbra{z}
\end{align}

Property \ref{item:property_monotonicity} is equivalent to that $S^{(j)}_{\mathsf{M}',\gamma}(\rho) -S^{(j)}_{\mathsf{M},\gamma}(\rho) \geq 0$. For $j=\numU,\numBS$, they have the same form:
\begin{align}
    &\phantom{=} S^{(j)}_{\mathsf{M}',\gamma}(\rho) -S^{(j)}_{\mathsf{M},\gamma}(\rho) \nonumber\\
    &= D(\mM(\rho)\|\mM(\gamma)) -D((\mathcal{P}\circ\mM)(\rho) \| (\mathcal{P}\circ\mM)(\gamma)) \nonumber\\
    & \geq 0
\end{align}
The inequality is due to the data-processing inequality of relative entropy.

\section{Proof of \cref{eq:ineq_12}} \label{app:proof_ineq_12}
By definition of $Q_F$ and $Q_R^\gamma$ in Eqs.~\eqref{eq:Q_F} and \eqref{eq:Q_R},
\begin{align}
    Q_F &= \sum_y \ketbra{y} \otimes \sqrt{\rho^T}\Pi_y^T\sqrt{\rho^T} \\
    Q_R &=  \sum_y \ketbra{y} \otimes \frac{\Tr[\Pi_y \rho]}{\Tr[\Pi_y \gamma]} \sqrt{\gamma^T}\Pi_y^T\sqrt{\gamma^T}
\end{align}
and thus
\begin{align}
    D(Q_F\|Q_R) &= \Tr\left[\sum_y \ketbra{y} \otimes \sqrt{\rho^T}\Pi_y^T\sqrt{\rho^T} \left(\ln\sqrt{\rho^T}\Pi_y^T\sqrt{\rho^T} - \ln\frac{\Tr[\Pi_y \rho]\sqrt{\gamma^T}\Pi_y^T\sqrt{\gamma^T}}{\Tr[\Pi_y \gamma]}\right) \right] \nonumber \\
    &=\sum_y  \Tr\left[ \sqrt{\rho}\Pi_y\sqrt{\rho} \left( \ln\frac{\sqrt{\rho}\Pi_y\sqrt{\rho}}{\Tr[\Pi_y \rho]} - \ln\frac{\sqrt{\gamma}\Pi_y\sqrt{\gamma}}{\Tr[\Pi_y \rho]} - (\ln \Tr[\Pi_y \rho] - \ln \Tr[\Pi_y \gamma])\openone \right) \right] \nonumber \\
    & = \sum_y D\left(\frac{\sqrt{\rho}\Pi_y\sqrt{\rho}}{\Tr[\Pi_y \rho]}\middle\| \frac{\sqrt{\gamma}\Pi_y\sqrt{\gamma}}{\Tr[\Pi_y \rho]} \right) - \sum_y\Tr[\sqrt{\rho}\Pi_y\sqrt{\rho}](\ln \Tr[\Pi_y\rho] - \ln\Tr[\Pi_y\gamma]) \nonumber \\
    & \geq D\left(\sum_y \frac{\sqrt{\rho}\Pi_y\sqrt{\rho}}{\Tr[\Pi_y \rho]}\middle\| \sum_y \frac{\sqrt{\gamma}\Pi_y\sqrt{\gamma}}{\Tr[\Pi_y \rho]} \right) - D(\mathcal{M}(\rho) \| \mathcal{M}(\gamma)) \nonumber \\
    & = D(\rho\|\gamma) - D(\mathcal{M}(\rho) \| \mathcal{M}(\gamma)) \label{eq:proof_ineq_12}
\end{align}
where the inequality comes from the joint convexity of the Umegaki relative entropy. Similar inequalities hold for any other relative entropies that satisfy the joint convexity, such as the Belavkin-Staszewski one. Adding $S(\rho)$ to both sides of \cref{eq:proof_ineq_12} gives \cref{eq:ineq_12}.

\end{document}